    \documentclass[10pt, a4paper, twocolumn]{IEEEtran}

\usepackage{amsmath,amsthm,epsfig,amssymb,verbatim,amsopn,cite,subfigure,multirow}
\usepackage{balance,nopageno}
\usepackage[usenames,dvipsnames]{color}
\usepackage[all]{xy}  
\usepackage{url}
\usepackage{amsfonts}
\usepackage{amssymb}
\usepackage{epsfig}
\usepackage{epstopdf}
\usepackage{slashbox}
\usepackage{bm}
\usepackage{balance}
\usepackage[margin=16.0mm,top=18mm,bottom=18mm]{geometry}

\newtheorem{Lemma}{Lemma}

\newtheorem{Corollary}[Lemma]{Corollary}

  {\proof}{\proofend}
\newtheorem{proposition}{Proposition}

\newcommand{\ve}[1]{\boldsymbol{#1}}
\newcommand{\E}[1]{\mathbb{E}\left\{#1\right\}}

 \newcommand{\vh}{\ve{h}}

 \newcommand{\vw}{\ve{w}}

\newcommand{\Us}{\mathsf{u}}
\newcommand{\Ds}{\mathsf{d}}
\newcommand{\AP}{\mathsf{a}}

\newcommand{\SINRd}{\mathsf{SINR_d}}
\newcommand{\SINRAP}{\mathsf{SINR_a}}

\newcommand{\SINRi}{\mathsf{SINR}_i}

\newcommand{\SNRd}{\mathsf{SNR_d}}
\newcommand{\SNRa}{\mathsf{SNR_a}}
\newcommand{\MRC}{\mathsf{MRC}}
\newcommand{\MRT}{\mathsf{MRT}}
\newcommand{\RFD}{R_\mathsf{{FD}}}
\newcommand{\RHDs}{R_\mathsf{{HD}}^{\mathsf{{RC}}}}
\newcommand{\RHDd}{R_\mathsf{{HD}}^\mathsf{{AC}}}
\newcommand{\snr}{\mathsf{snr}}
\newcommand{\HD}{\mathsf{HD}}
\newcommand{\FD}{\mathsf{FD}}

\newcommand{\EI}[1]{\mathbb{E}_{I_{\Ds,\Us}}\left\{#1\right\}}

\newcommand{\Erth}[1]{\mathbb{E}_{r,\theta}\left\{#1\right\}}
\newcommand{\Er}[1]{\mathbb{E}_{r}\left\{#1\right\}}

\newcommand{\Sn}{\sigma_n^2}
\newcommand{\Sap}{\sigma_{\AP\AP}^2}

\newcommand{\Prob}{\textnormal{Pr}}

\newcommand{\AuthorOne}{Mohammadali Mohammadi$^\dag$}
\newcommand{\AuthorTwo}{Himal A. Suraweera$^\S$}
\newcommand{\AuthorThree}{Ioannis Krikidis$^\ddag$}
\newcommand{\AuthorFour}{Chintha Tellambura$^*$}

\usepackage{multibib}
\usepackage[nodisplayskipstretch]{setspace}
\newcites{Prim}{Very important papers}

\definecolor{light-gray}{gray}{0.65}

\newcounter{mytempeqcounter}

\newcommand{\ThankFour}{Part of this work was supported by the Research Promotion Foundation, Cyprus under the project KOYLTOYRA/BP-NE/0613/04 ``Full-Duplex Radio: Modeling, Analysis and Design (FD-RD)''.}

\title{Full-Duplex Radio for Uplink/Downlink Transmission with Spatial Randomness}

\author{\authorblockN{\AuthorOne,\:\AuthorTwo,\: \AuthorThree,\: and \AuthorFour}\\
\small{
$^\dag$Faculty of  Engineering, Shahrekord University, Iran (e-mail: m.a.mohammadi@eng.sku.ac.ir)\\
$^\S$Department of Electrical and Electronic Engineering, University of Peradeniya, Sri Lanka (e-mail: himal@ee.pdn.ac.lk)\\
$^\ddag$Department of Electrical and Computer Engineering,University of Cyprus, Cyprus (e-mail: krikidis@ucy.ac.cy)\\
$^*$Department of Electrical and Computer Engineering, University of Alberta, Canada (e-mail: chintha@ece.ualberta.ca)
}}\normalsize

\begin{document}

\maketitle
\thispagestyle{empty}

\vspace{-1em}
\begin{abstract}
We consider a wireless system with a full-duplex (FD) access point (AP) that transmits to a scheduled user in the downlink (DL) channel, while receiving data from an user in the uplink (UL) channel at the same time on the same frequency. In this system, loopback interference (LI) at the AP and inter user interference between the uplink (UL) user and  downlink (DL)  user can cause performance degradation. In order to characterize the effects of LI and inter user interference, we derive closed-form expressions for the outage probability and achievable sum rate of the system. In addition an asymptotic analysis that reveals insights into the system behavior and performance degradation is presented. Our results indicate that under certain conditions, FD transmissions yield  performance gains over half-duplex  (HD)  mode of operation.
\let\thefootnote\relax\footnotetext{\ThankFour}
\end{abstract}
\vspace{-0.5em}
\section{Introduction}
Due to the exponential growth of wireless traffic, spectral efficiency improvements achievable from   transmitting  while receiving are highly beneficial~\cite{Margetts:2007,Sabharwal:JSac:2014}. Traditionally, this was achieved by the separation of the transmit and receive carrier frequency. However, if a wireless radio node can only transmit or receive at a given time and frequency, a loss of efficiency  from a channel resource perspective must be expected. A promising solution that can be employed to avoid the loss of spectral efficiency is the full-duplex (FD) technology\cite{Duarte:PhD:dis,Riihonen:spawc:2009,Riihonen:TSP:2011,Khojastepour:Mobicom:2012,Sachin:NSDI:2014}.

Since the loopback interference (LI) caused by a node that is both transmitting and receiving at the same time can be overwhelming, up until now FD operation was considered practically unrealistic. This perception has been challenged due to the recent advances in antenna design and analog/digital signal processing. To this end, several recent works have described single and multiple antenna FD system designs largely made possible through new LI cancellation techniques~\cite{Duarte:PhD:dis,Khojastepour:Mobicom:2012,Ngo:JSAC:2014}. The implementation of single antenna FD technology with LI cancellation was demonstrated in~\cite{Duarte:PhD:dis}. A multiple-input multiple-output (MIMO) FD implementation (MIDU) was presented in~\cite{Khojastepour:Mobicom:2012}, while~\cite{Sachin:NSDI:2014} reported design and implementation of an in-band WiFi-PHY based FD MIMO system. In~\cite{Ngo:JSAC:2014} a massive MIMO FD relay system with spatial LI mitigation and optimum power allocation was investigated.

An interesting application of FD communications is simultaneous uplink and downlink transmission in wireless systems such as WiFi and cellular networks~\cite{Sanjay:CISS:2013,Girnyk:2013,Sachin:NSDI:2014}. However, such transmissions introduce LI and internode interference in the network as downlink transmission will be affected by the LI and the uplink user will interfere with the downlink reception. Therefore, in the presence of such interference, it is not clear whether FD applied to uplink/downlink user settings can bring performance benefits. In order to answer this question, several works in the literature have presented useful results. In~\cite{Sanjay:CISS:2013} a FD cellular analytical model based on stochastic geometry was used to derive the sum capacity of the system. However,~\cite{Sanjay:CISS:2013} assumed perfect LI cancellation and therefore, the effect of LI is not included in the results. In~\cite{DBeiYin:ACSSC} the combination of FD and  massive MIMO was considered for simultaneous uplink/downlink cellular communication. The information theoretic study presented in~\cite{Achaleshwar:acssc:DS13}, has investigated the rate gain achievable in a FD uplink/downlink network with internode interference management techniques. The application of FD radios for a single small cell scenario was considered in~\cite{Panwar:ICC14}. Specifically in this work, the conditions where FD operation provides a throughput gain compared to HD and the corresponding throughput results using simulations were presented. In~\cite{tspNguyen:2013}, joint precoder designs to optimize the spectral and energy efficiency of a FD multiuser MIMO system were presented. However~\cite{Achaleshwar:acssc:DS13,DBeiYin:ACSSC, tspNguyen:2013} considered fixed user settings for performance analysis and as such the effect of interference due to distance, particularly relevant for wireless networks with spatial randomness, is ignored.

In this paper, we consider a wireless network scenario in which a FD infrastructure node is communicating with half-duplex (HD) spatially random user terminals to support simultaneous uplink and downlink transmissions. Our contributions are summarized as follows:
\begin{itemize}
\item We take both LI and inter user inference into account and derive exact expressions for the outage probability and achievable sum rate of the system. Moreover, to highlight the system  behavior  and shed insights into the performance degradation, an asymptotic analysis is also presented.

\item We have compared the sum rate performance of the system for FD and HD modes of operation at the AP to elucidate the signal-to-noise ratio regions where the former mode of operation outperforms the latter mode of operation. Moreover, our results indicate that different power levels at the AP and UL user has a significant adverse effect to lower the sum rate in the HD mode of operation than the FD counterpart.
\end{itemize}
\section{System Model}\label{sec:system model and assumption}
Consider a single cell wireless system with an access point (AP), where data to the users in the DL channel, and data from users in the UL channel are transmitted and received at the same time on the same frequency.
All users are located in a circular area with radius $R_c$ and the AP is located at the center.
We assume that users are equipped with a single antenna, while the AP is equipped with two antennas (one antenna is used to transmit in the DL channel while the other antenna is used for UL channel reception). In the sequel we use subscript-$\Us$ for the UL user, subscript-$\Ds$ for the DL user, and subscript-$\AP$ for the AP. Similarly, we will use subscript-$\AP\AP$, subscript-$\AP{\Ds}$, subscript-${\Us}{\Ds}$, and subscript-${\Us}\AP$ to denote the AP-to-AP, AP-to-DL user, UL user-to-DL user, and UL user-to-AP channels, respectively.

Let $\Phi_{\Ds}$ be a two-dimensional homogeneous Poisson point process (PPP) with density $\lambda_{\Ds}$ that characterizes the spatial distribution of the DL users over $\mathbb{R}^2$. To obtain the most essential features, we consider the widely used Poisson bipolar model~\cite{Blaszczyszyn:book:2010} and assume that the UL users are located at a fixed distance $d$ in a random direction of angle $\theta$ from the DL users. The results obtained thus can be interpreted as the performance of networks with random link distances conditioned on the link distance having a certain value. The AP selects a DL user that is physically nearest to it. We use the terms ``nearest DL user'' and ``scheduled DL user'' interchangeably throughout the paper to refer to this user. In next generation ultra-dense networks, each user will be in the coverage area of an AP and can be considered as a most nearest user~\cite{Andrews:MCOM:2013}. Selection of a nearest user also serves as a practical consideration for FD implementation since transmitting very high power signals towards distant periphery users in order to guarantee a quality-of-service can cause overwhelming LI at the receive side of the AP. Moreover, as a benchmark comparison we also consider the random user selection (RUS) in Section~\ref{sec:Numerical results}. In RUS method the AP randomly selects one of all candidate DL users with equal probability.

We assume that the links in the network experience both large-scale path loss effects and small-scale Rayleigh fading phenomenon. For the large-scale path loss, we assume the standard singular path loss model, $\ell(x,y)=\|x-y\|^{-\alpha}$, where $\alpha \geq 2$ denotes the path-loss exponent and $\|x-y\|$ is the Euclidean distance between two nodes.

The received power at a typical DL user located at point $x_{\Ds}$ from the AP is $P_{\AP} h_{\AP{\Ds}} \ell(x_{\Ds})$. It is worth mentioning that the scheduled UL user, located at $x_{\Us}$, is served by receive antenna from AP at the same time, and it lacks coordination with concurrent active DL users. Therefore, the signal-to-interference-plus-noise ratio (SINR) of the typical DL user associated with the AP can be expressed as
\begin{align}\label{eq:SINR: downlonk user:single antenna case}
\SINRd=\frac{P_{\AP} h_{\AP\Ds}  \ell(x_{\Ds})}
{P_{\Us} h_{ \Us\Ds} \ell( x_{\Us},x_{\Ds}) + \Sn},
\end{align}
where $P_{\Us}$ denotes the transmit power of the UL user in UL channel and $\Sn$ is the constant additive noise power. On the other hand, received power at the AP from the active UL user is $P_{\Us} h_{\Us\AP} \ell(x_{\Us} )$. Due to the FD mode of operation, the receive antenna of the AP will receive a LI from its transmit antenna. Hence, the resulting SINR expression at the AP can be written as
\begin{align}\label{eq:SINR at AP in uplink}
\SINRAP=
\frac{P_{\Us} h_{\Us\AP}  \ell(x_{\Us} )}
{P_{\AP} h_{\AP\AP}  + \Sn},
\end{align}
where $h_{\AP\AP}$ denotes the LI channel at the AP. In order to mitigate the adverse effects of self-interference on system performance, an interference cancellation scheme (i.e. analog/digital cancellation) can be used at the AP and we model the residual LI channel with Rayleigh fading assumption since the strong line-of-sight component can be estimated and removed~\cite{Riihonen:spawc:2009,Riihonen:TSP:2011,Himal:FD:JWCOM}. Since each implementation of a particular analog/digital LI cancellation scheme can be characterized by a specific residual power, a parameterization by $h_{\AP\AP}$  satisfying $\E{|h_{\AP\AP}|^2}=\Sap$ allows these effects to be studied in a generic way~\cite{Riihonen:TWC:2011,Himal:FD:JWCOM}.

In order to facilitate the ensuing analysis, we now set up a polar coordinate system in which the origin is at the AP and the scheduled DL user is at $x_{\Ds}=(r,0)$. Therefore, according to the bipolar poisson model, we have $\ell(x_{\Us}) = (r^2+d^2-2rd\cos\theta)^{-\alpha/2}$.
In the following, we will need the exact knowledge of the spatial distribution of the $\ell( x_{\Us})$ in terms of $r$ and $\theta$. Since we assume that nearest DL user is scheduled for downlink transmission, $x_{\Ds}$ denotes the distance between the AP and the nearest DL user. Therefore, the probability distribution function (pdf) of the nearest distance $x_{\Ds}$ for the homogeneous
PPP $\Phi_{\Ds}$ with intensity $\lambda_{\Ds}$ is given by~\cite{Haenggi:IT:2005}
\vspace{-0.2em}
\begin{align}\label{eq:pdf of the nearst distance}
 f_{r}(r)= 2\pi\lambda_{\Ds}r e^{-\lambda_{\Ds}\pi r^2},~\qquad r\geq0.
\end{align}
Moreover, angular distribution is uniformly distributed over $[0~2\pi]$ i.e., $f_{\theta}(\theta )=1/ 2\pi$.
\section{Performance Analysis}
In this section, we derive analytical outage probability and sum rate expressions. First, we obtain the cumulative distribution function (cdf) of the SINRs, $\SINRd$ and $\SINRAP$. Next exploiting the cdf result, the outage probability and sum rate are derived.
\vspace{-1 em}
\subsection{The SINR cdfs at the AP and DL User}
The cdf of the $\SINRAP$ and the $\SINRd$ are respectively expressed by
\vspace{-0.3em}
\begin{align}\label{eq:cumulative distribution function of SINR}
F_{\SINRi}(z) = 1-\Prob(\SINRi \geq z),
\end{align}
for $i\in\{\AP, \Ds\}$ and $z\geq0$, where $\Prob(\cdot)$ denotes the probability. We now proceed to derive exact expressions for $F_{\SINRAP}(z)$ and $F_{\SINRd}(z)$, respectively.
 \newline \newline
\textbf{\emph{{Uplink Transmission:}}}
Using~\eqref{eq:cumulative distribution function of SINR}, the $\SINRAP$ cdf can be written as
\vspace{-0.5em}
\begin{align}\label{eq:cumulative distribution function of SINRu}
&F_{\SINRAP}(z)
\!= 1\!-\!  \Erth{\Prob \left(h_{\Us\AP}\! \geq\!
\frac{z}{P_{\Us} \ell(x_{\Us})}\! [P_{\AP} h_{\AP\AP} \!+\!\Sn] \Big| h_{\AP\AP}\right)}
\nonumber\\
&\quad\!=\!1\!-\!
\Erth{
\frac{e^{-z\frac{\Sn}{P_{\Us}}(r^2+d^2-2rd\cos\theta)^{\alpha/2} }}
{1 +z\frac{P_{\AP}}{P_{\Us}}\Sap (r^2\!+\!d^2\!-\!2rd\cos\theta)^{\alpha/2} }}\!,
 \end{align}
where the second equality in \eqref{eq:cumulative distribution function of SINRu} is due to $h_{\AP\AP}\sim \exp(1/\Sap)$. With the aid of the pdfs for $r$ and $\theta$, we can express $F_{\SINRAP}(z)$ as $F_{\SINRAP}(z)=$
\vspace{-0.2em}
\begin{align}\label{eq: cdf of SINRd polar coordinates general}
1-
\lambda_{\Ds}
\int_{0}^{R_c}
\int_{0}^{2\pi}\frac{r e^{-\lambda_{\Ds}\pi r^2}e^{-\frac{z\Sn}{P_{\Us}} (r^2+d^2-2rd\cos\theta)^{\frac{\alpha}{2}}}}
{1 +z \frac{P_{\AP}}{P_{\Us}} \Sap(r^2+d^2-2rd\cos\theta)^{\frac{\alpha}{2}}}d\theta dr.
\end{align}
In general, the double integral in~\eqref{eq: cdf of SINRd polar coordinates general} does not admit a simple analytical solution for an arbitrary value of $\alpha$. However, the cdf can be conveniently evaluated using numerical integration. The following propositions characterize $F_{\SINRAP}(z)$ for the interference-limited scenario with $\Sn=0$ and special cases\footnote{Note that $\alpha=2$ and $\alpha=4$ correspond to free space propagation and typical rural areas, respectively, and constitute useful bounds for practical propagation conditions.}; $\alpha=2$ and $\alpha=4$.
\begin{proposition}
The cdf of $\SINRAP$, for $\alpha=2$ is given by
\vspace{-0.3em}
\begin{align}\label{eq: cdf of SINRd integral over r: alpha 2 Final}
&F_{\SINRAP}(z)
=1-\frac{P_{\Us}}{P_{\AP }} \frac{ 8\pi\lambda_{\Ds} }{z\Sap}
~\sum_{k=0}^{\infty}\frac{(-2\pi\lambda_{\Ds} c)^k}{\Gamma(k+1)}
\sqrt{c}\left(\frac{b-\sqrt{c}\varrho}{c - b^2}
\right)^{k+1}\nonumber\\
&~\times
F_{1}\left(\!k+1;\!k+1,\!k+1;\!k+2;\!\frac{b\!-\!\sqrt{c}\varrho}{b+\sqrt{c}},\frac{b\!-\!\sqrt{c}\varrho}{b-\sqrt{c}}\right),
\end{align}
where $c =\left(\frac{P_{\Us}}{P_{\AP}} \frac{ 1}{z\Sap}+d^2\right)^2 $, $b =\frac{P_{\Us}}{P_{\AP}} \frac{ 1}{z\Sap}-d^2$,  $\varrho=(\sqrt{R_c^4+bR_c^2+c}-\sqrt{c})/R_c^2$, $\Gamma(\cdot)$ is the Gamma function~\cite[Eq. (8.310.1)]{Integral:Series:Ryzhik:1992}, and $F_1(\cdot;\cdot,\cdot;\cdot;\cdot,\cdot)$ is the Appell hypergeometric function~\cite[Eq. (5.8.5)]{Transcendental:book}.
\end{proposition}

\begin{proof}
Following~\eqref{eq: cdf of SINRd polar coordinates general}, the $F_{\SINRd}(z)$ corresponding to $\alpha=2$ and $\Sn=0$ is given by $F_{\SINRAP}(z)=$
\vspace{-0.3em}
\begin{align}
1-\frac{P_{\Us}}{P_{\AP}} \frac{ 1}{z\Sap}
\int_{0}^{R_c}
\int_{0}^{2\pi}\frac{\lambda_{\Ds}r e^{-\lambda_{\Ds}\pi r^2}}
{\frac{P_{\Us}}{P_{\AP}} \frac{ 1}{z\Sap} + r^2+d^2-2rd\cos\theta}d\theta dr.\nonumber
\end{align}
With the help of~\cite[Eq. (3.661.4)]{Integral:Series:Ryzhik:1992}, and next making the change of variable $r^2=\upsilon$, we obtain
\vspace{-0.3em}
\begin{align}\label{eq: cdf of SINRd integral over r}
&F_{\SINRAP}(z)
=\!\!1-\frac{P_{\Us}}{P_{\AP }} \frac{ \pi\lambda_{\Ds} }{z\Sap}
\int_{0}^{R_c^2}\!\!
\frac{ e^{-\lambda_{\Ds}\pi \upsilon}}
{\sqrt{ \upsilon^2 + 2 b\upsilon + c}}d\upsilon.
\end{align}
To the best of our knowledge, the integral in~\eqref{eq: cdf of SINRd integral over r} does not admit a
closed-form solution. In order to proceed, we use Taylor series representation~\cite[Eq. (1.211.1)]{Integral:Series:Ryzhik:1992} for term $e^{-\lambda_{\Ds}\pi \upsilon}$, and write
\vspace{-0.3em}
\begin{align}\label{eq: cdf of SINRd integral over r: Taylor series}
F_{\SINRAP}(z)
&=1-\frac{P_{\Us}}{P_{\AP }} \frac{ \pi\lambda_{\Ds} }{z\Sap}\\
&~\times\sum_{k=0}^{\infty}\frac{(-\lambda_{\Ds}\pi)^k}{k!}
\!\!\int_{0}^{R_c^2}\!
\frac{  \upsilon^k}
{\sqrt{ \upsilon^2 + 2 b\upsilon + c}}d\upsilon.\nonumber
\end{align}
A change of variable $\sqrt{ \upsilon^2 + 2 b\upsilon + c}=\upsilon t+ \sqrt{c}$, and after some manipulations,~\eqref{eq: cdf of SINRd integral over r: Taylor series} can be expressed as
\vspace{-0.3em}
\begin{align}\label{eq: cdf of SINRd integral over r: Euler substitutions}
&F_{\SINRAP}(z)
=1-\frac{P_{\Us}}{P_{\AP }} \frac{ 4\pi\lambda_{\Ds} }{z\Sap}\nonumber\\
&\qquad\qquad\times\sum_{k=0}^{\infty}\frac{(-\lambda_{\Ds}\pi)^k}{k!}
\!\!\int_{\frac{b}{\sqrt{c}}}^{\varrho}\!
\frac{(b-\sqrt{c}t)^{k}}{ (t^2-1)^{k+1}}
dt.
\end{align}
Finally, using~\cite[Eq. (5.8.5)]{Transcendental:book}, we get the desired result given in~\eqref{eq: cdf of SINRd integral over r: alpha 2 Final}.
\end{proof}
\begin{proposition}\label{Prop:cdf of the SINRup alpha4}
The cdf of $\SINRAP$ for $\alpha=4$ is lower bounded as
\vspace{-0.3em}
\begin{align}\label{eq: proof of the cdf of SINRd alpha final}
&F_{\SINRAP}(z)
> 1-
 \sum_{k=0}^{\infty}\frac{(-1)^k(\lambda_{\Ds}\pi R_c^2)^{k+1}}{\Gamma(k+2)}\nonumber\\
&\qquad\times
{}_{2}F_{1} \left(\!1,\frac{k\!+\!1}{2},\frac{k\!+\!1}{2}+1,\!-z\Sap\frac{ P_{\AP}}{P_{\Us}}R_c^4\!\right)\!,
\end{align}
where ${}_{2}F_{1}(\cdot,\cdot;\cdot;\cdot)$ denotes the Gauss hypergeometric function defined in~\cite[Eq. (9.111)]{Integral:Series:Ryzhik:1992}.
\end{proposition}
\begin{proof}
Following~\eqref{eq: cdf of SINRd polar coordinates general}, the $F_{\SINRAP}(z)$ corresponding to $\alpha=4$ and $\Sn=0$ can be written as
\vspace{-0.3em}
\begin{align}\label{eq: proof of the cdf of SINRd alpha 4 over r and theta}
&F_{\SINRAP}(z)
 =1-\frac{ 1}{z}
 \frac{ \lambda_{\Ds}P_{\Us}}{\Sap P_{\AP}}
\\
&~\times
\!\!\int_{0}^{R_c}
\!\int_{0}^{2\pi}\!\frac{re^{-\lambda_{\Ds}\pi r^2}}
{\frac{P_{\Us}}{P_{\AP}} \frac{ 1}{z\Sap}\! +\!(r^2+d^2-2rd\cos\theta)^{2}}d\theta dr\!.\nonumber
\end{align}
By using~\cite{Integral:Series:Ryzhik:1992}, the inner integral can be obtained as
\begin{align}\label{eq: proof of the cdf of SINRd alpha 4 over r}
&F_{\SINRAP}(z)
=1-\frac{ \sqrt{2}\pi}{z}
 \frac{ \lambda_{\Ds}P_{\Us}}{\Sap P_{\AP}}\\
&~\times
\!\!\!\int_{0}^{R_c}\!\!\!
\frac{re^{-\lambda_{\Ds}\pi r^2}}
{\sqrt{c_2(r) \!\!+ \!\!\sqrt{c_4(r)c_0(r)}}}
\left[\frac{1}{\sqrt{c_0(r)}}\!\! + \!\!\frac{1}{\sqrt{c_4(r)}}\right] dr\!,\nonumber
\end{align}
where $c_0(r)=b_0(r)-b_1(r)+b_2(r)$, $c_2(r)=b_0(r)-b_2(r)$, and $c_4(r)=b_0(r)+b_1(r)+b_2(r)$, with $b_0(r)=P_{\Us}/( P_{\AP}z\Sap) + (r^2+d^2)^2$, $b_1(r)=4rd(r^2+d^2)$, and $b_2(r)=4r^2d^2$. The integral in~\eqref{eq: proof of the cdf of SINRd alpha 4 over r} cannot be calculated analytically. However, we can simplify the above integral in the case of $d=0$. Hence, after a simple substitution $r^2=\upsilon$,~\eqref{eq: proof of the cdf of SINRd alpha 4 over r} can be written as
\begin{align}\label{eq: proof of the cdf of SINRd alpha 4 over r simplified}
&F_{\SINRAP}(z)
> 1-\frac{ \pi}{z}
 \frac{ \lambda_{\Ds}P_{\Us}}{\Sap P_{\AP}}
\int_{0}^{R_c^2}
\frac{e^{-\lambda_{\Ds}\pi \upsilon}}
{\upsilon^2+\frac{ P_{\Us}}{ P_{\AP}} \frac{1}{z\Sap}}d\upsilon.
\end{align}
In order to simplify \eqref{eq: proof of the cdf of SINRd alpha 4 over r simplified}, we adopt a series expansion of the
exponential term. Substituting the series expansion of $e^{-\lambda_{\Ds}\pi \upsilon}$ into the~\eqref{eq: proof of the cdf of SINRd alpha 4 over r simplified} yields
\begin{align}\label{eq: proof of the cdf of SINRd alpha 4 over r taylor}
&F_{\SINRAP}(z)
> 1-\frac{ 1}{z}
 \frac{ P_{\Us}}{\Sap P_{\AP}}\times\nonumber\\
 &\qquad\qquad
 \sum_{k=0}^{\infty}\frac{(-\lambda_{\Ds}\pi)^{k+1}}{k!}
\int_{0}^{R_c^2}
\frac{ \upsilon^k}
{\upsilon^2+\frac{ P_{\Us}}{ P_{\AP}} \frac{1}{z\Sap}}d\upsilon.
\end{align}
Let us denote $\beta=\frac{P_{\Us}}{ P_{\AP}} \frac{1}{z\Sap}$.
By making the change of variable $\left(\upsilon/R_c^2\right)^2=t$, we obtain
\begin{align}\label{eq: proof of the cdf of SINRd alpha 4 over r taylor}
F_{\SINRAP}(z)
&\!> \!1\!-\!
 \sum_{k=0}^{\infty}\frac{(-\lambda_{\Ds}\pi R_c^2)^{k+1}}{2k!}
\!\!\int_{0}^{1}\!\!
\frac{ t^{\frac{k-1}{2}}}
{1\!+\frac{R_c^4}{\beta}t}dt.
\end{align}
Now with the help of~\cite[Eq. (9.111)]{Integral:Series:Ryzhik:1992} the integral in \eqref{eq: proof of the cdf of SINRd alpha 4 over r taylor} can be solved to yield \eqref{eq: proof of the cdf of SINRd alpha final}.
\end{proof}

\textbf{\emph{Downlink Transmission:}}
Using~\eqref{eq:SINR: downlonk user:single antenna case} and \eqref{eq:cumulative distribution function of SINR}, the cdf of $\SINRd$ can be written as
\begin{align}\label{eq:cdf of SINRd first step}
&F_{\SINRd}(z)
=1-\EI{\Prob\left(P_{\AP} h_{\AP\Ds} r^{-\alpha}
 \geq z [I_{\Ds,\Us} + \Sn]\right)\big\vert r},\nonumber\\
&\qquad\qquad=1-
\EI{e^{-\frac{z}{P_{\AP}}r^{\alpha} [I_{\Ds,\Us} +\Sn]}\big\vert r}.
\end{align}
Note that in our system model the randomness of the $I_{\Ds,\Us}$ is due to the fading power envelope $h_{\Ds \Us}$. As such, $F_{\SINRd}(z)$ can be written as
\begin{align}\label{eq:cdf of SINRd sec step}
F_{\SINRd}(z)
&=1-\Er{e^{-\frac{z}{P_{\AP}}\Sn r^{\alpha} }
\int_{0}^{\infty}
 e^{-\left(\frac{r}{d}\right)^{\alpha}\frac{P_{\Us}}{P_{\AP} }zx}   e^{-x}dx},\nonumber\\
&=1-
2\pi\lambda_{\Ds}\int_{0}^{R_c}r\frac{e^{-z\frac{\Sn}{P_{\AP}}r^{\alpha} } e^{-\lambda_{\Ds}\pi r^2} }
{1 +\left(\frac{r}{d}\right)^{\alpha}\frac{P_{\Us}}{P_{\AP} }z}dr.
\end{align}
Eq.~\eqref{eq:cdf of SINRd sec step} does not have a closed-form solution. However, an expression for $F_{\SINRd}(z)$ can be derived in the interference-limited case  in Proposition~\ref{Prop:cdf of the SINRd}.
\begin{proposition}\label{Prop:cdf of the SINRd}
The cdf of $\SINRd$, can be expressed as
\vspace{-0.2em}
\begin{align}\label{eq: final cdf of SIR DL}
&F_{\SINRd}(z)
= 1-
 \sum_{k=0}^{\infty}\frac{(-1)^k(\lambda_{\Ds}\pi R_c^2)^{k+1}}{\Gamma(k+2)}\\
&\quad\times
{}_{2}F_{1} \left(\!1,\frac{2(k\!+\!1)}{\alpha},\frac{2(k\!+\!1)}{\alpha}+1,\!-z\frac{P_{\Us}}{ P_{\AP}}\left(\frac{R_c}{d}\right)^{\alpha}\right)\!.\nonumber
\end{align}
\end{proposition}
\begin{proof}
The proof, similar to \emph{Proposition~\ref{Prop:cdf of the SINRup alpha4}}, is omitted.
\end{proof}
\vspace{-1.2em}
\subsection{Outage Probability}
The outage probability is an important quality-of-service metric defined
as the probability that $\SINRi$, $i\in\{\AP, \Ds\}$,
drops below an acceptable SINR threshold, $\gamma_{th}$. We now present the following corollaries to establish the DL and UL user outage
probability valid in the interference-limited case (i.e., $\Sn=0$).
\begin{Corollary}
The UL user outage probability with $\alpha=2$ is given by substituting $z=\gamma_{th}$ into~\eqref{eq: cdf of SINRd integral over r: alpha 2 Final}. Moreover, for  $\alpha=4$, the outage probability is lower bounded by substituting $z=\gamma_{th}$ into~\eqref{eq: proof of the cdf of SINRd alpha final}.
\end{Corollary}
\begin{Corollary}
The UL user outage probability is given by substituting $z=\gamma_{th}$ into~\eqref{eq: final cdf of SIR DL}.
\end{Corollary}
\vspace{-0.9em}
\subsection{Achievable Sum Rate}
The achievable sum rate with simultaneous UL/DL transmission can be written as
\vspace{-0.2em}
\begin{align}
\RFD  = R_{\AP} + R_{\Ds},\label{eq:achievable rate FD}
\end{align}
where $R_{\AP}=\E{ \log_2 \left[1+\SINRAP\right]}$ and  $R_{\Ds}=\E{ \log_2 \left[1+\SINRd\right]}$ are the spatial
average capacity of the UL ($x_{\Us}\rightarrow \mathrm{AP }$) and DL ($\mathrm{AP} \rightarrow x_{\Ds}$), respectively.

Note that since $\E{X} = \int_{t=0}^{\infty}\Prob(X>t)dt$ for a nonnegative random variable $X$, the spatial
average capacity can be written as
\vspace{-0.1em}
\begin{align}
R_i&=\int_{0}^{\infty}\left[1-F_{\SINRi}(\epsilon_t)\right]dt.\label{eq:achievable downlink rate}
\end{align}
where $i\in\{\AP, \Ds\}$ and $\epsilon_t=2^t-1$.

\textbf{\emph{\\Uplink Transmission:} }
By substituting \eqref{eq:cumulative distribution function of SINRu} into~\eqref{eq:achievable downlink rate}, the exact average capacity of the UL user can be written as
\begin{align}\label{eq:spatial average capacity of the downlink user}
&R_{\AP}=
\int_{0}^{\infty}\!\!\int_{0}^{R_c}\!\!\int_{0}^{2\pi}\!\!\\
&~
\frac{2\pi\lambda_{\Ds}re^{-\epsilon_t\frac{\Sn}{P_{\Us}}(r^2+d^2-2rd\cos\theta)^{\alpha/2} }e^{-\lambda_{\Ds}\pi r^2}}
{1 +\epsilon_t\frac{P_{\AP}}{P_{\Us}}\Sap     (r^2+d^2-2rd\cos\theta)^{\alpha/2} }dr d\theta dt.\nonumber
 \end{align}
This integral cannot be solved in closed-form. Therefore, we now turn our attention into deriving the average capacity of the UL user with the interference-limited assumption and $\alpha=2, 4$.

\begin{Corollary}
Plugging ~\eqref{eq: cdf of SINRd integral over r: alpha 2 Final} into~\eqref{eq:achievable downlink rate}, the spatial average capacity of the UL user for $\alpha=2$ is given by
\begin{align}\label{eq: acheivable rate UL: alpha 2 Final}
&R_{\AP}\!= \!\!\frac{8\pi\lambda_{\Ds} }{\Sap\log2}\frac{P_{\Us}}{P_{\AP }}
\!\sum_{k=0}^{\infty}\!
\frac{(-2\pi\lambda_{\Ds})^k}
{\Gamma(k+1)}
\!\!\int_0^{\infty}\!\!\!\!\!
\frac{ c^{k+\frac{1}{2}}}{z(z+1)}
\left(\!\frac{b\!-\!\sqrt{c}\varrho}{c \!-\! b^2}
\right)^{k+1}\nonumber\\
&~\times
F_{1}\left(\!k+1;\!k+1,\!k+1;\!k+2;\!\frac{b\!-\!\sqrt{c}\varrho}{b+\sqrt{c}},\frac{b\!-\!\sqrt{c}\varrho}{b-\sqrt{c}}\right)dz.
\end{align}
\end{Corollary}

\begin{proposition}\label{Prop: Acheivable rate uplink}
For $\alpha=4$, the spatial average capacity of the UL user is upper bounded by
\begin{align}\label{eq:achievable uplink rates:special case:proposition}
R_{\AP}&< \frac{2}{\log 2}\sum_{k=0}^{\infty}\frac{(-1)^k(\lambda_{\Ds}\pi R_c^2)^{k+1}}{(k+1)\Gamma(k+2)}\nonumber\\
&\quad\times
G_{3 ~3}^{2~ 3} \left(\frac{ P_{\AP}}{P_{\Us}}R_c^4\Sap \bigg\vert {0,1-\frac{k+1}{2},0 \atop 0, 0, -\frac{k+1}{2}} \right),
\end{align}
where \small{$G_{p q}^{m n} \left( z \  \vert \  {a_1\cdots a_p \atop b_1\cdots b_q} \right)$ }\normalsize denotes the Meijer G-function defined in~\cite[ Eq. (9.301)]{Integral:Series:Ryzhik:1992}.
\end{proposition}
\begin{proof}
By substituting the lower bound of $F_{\SINRAP}(\cdot)$ from Proposition~\ref{Prop:cdf of the SINRup alpha4} into~\eqref{eq:achievable downlink rate}, and applying the transformation $y=2^t-1$, an upper bound for the average capacity of the UL user can be derived as
\begin{align}
R_{\AP}&< \frac{1}{\log2}\sum_{k=0}^{\infty}\frac{(-1)^k(\lambda_{\Ds}\pi R_c^2)^{k+1}}{\Gamma(k+2)}\nonumber\\
&\times\underbrace{\int_0^{\infty} \frac{1}{y+1} {}_{2}F_{1} \left(\!1,\frac{k\!+\!1}{2},\frac{k\!+\!1}{2}+1,\!-\Sap\frac{ P_{\AP}}{P_{\Us}}R_c^4 y\!\right)dy,}_{\mathcal{I}}\nonumber
\end{align}
where the integral, $\mathcal{I}$ can be
expressed~\cite[Eq. (17)]{Adamchik:1990} in terms of the tabulated Meijer
G-function as
\begin{align}
\mathcal{I}= \frac{2}{k+1}&\int_0^{\infty}
G_{1 1}^{1 1} \left( y \ \Big \vert \  {0 \atop 0} \right)\nonumber\\
&\!\times
G_{2 2}^{1 2} \left( \frac{P_{\AP}}{P_{\Us}}R_c^4\Sap y\! \  \Big\vert \  {0, 1-\frac{k+1}{2} \atop 0,-\frac{k+1}{2}} \right)dy.
\end{align}
The above integral can be solved with the help of~\cite[Eq. (21)]{Adamchik:1990}  to yield the desired result in~\eqref{eq:achievable uplink rates:special case:proposition}.
\end{proof}

\textbf{\emph{{\\Downlink Transmission:}}}
By plugging \eqref{eq:cdf of SINRd sec step} into~\eqref{eq:achievable downlink rate}, exact average capacity of the DL user can be written as
\begin{align}\label{eq: ache. sum rate dl user general}
R_{\Ds}=2\pi\lambda_{\Ds}\int_{0}^{\infty}\int_{0}^{\infty}\frac{r e^{-\epsilon_t\frac{\Sn}{P_{\AP}}r^{\alpha} } e^{-\lambda_{\Ds}\pi r^2} }
{1 +\left(\frac{r}{d}\right)^{\alpha}\frac{P_{\Us}}{P_{\AP} }\epsilon_t}dr.
\end{align}
Moreover, for the interference-limited case (i.e., $\Sn=0$), using the cdf in~\eqref{eq: final cdf of SIR DL}, Proposition~\ref{Prop: Acheivable rate downlink} presents the average capacity of the DL user.
\begin{proposition}\label{Prop: Acheivable rate downlink}
The spatial average capacity of the DL user in the interference-limited case is  expressed as
\begin{align}\label{eq:achievable sum rate of the DL}
R_{\Ds}&= \frac{\alpha}{2}\sum_{k=0}^{\infty}\frac{(-1)^k(\lambda_{\Ds}\pi R_c^2)^{k+1}}{(k+1)\Gamma(k+2)}\nonumber\\
&\quad
\times G_{3 ~3}^{2~ 3} \left(\frac{P_{\Us}}{ P_{\AP}}\left(\frac{R_c}{d}\right)^{\alpha} \bigg\vert {0,1-\frac{2(k+1)}{\alpha},0 \atop 0, 0, -\frac{2(k+1)}{\alpha}} \right).
\end{align}
\end{proposition}
\begin{proof}
The proof, similar to \emph{Proposition~\ref{Prop: Acheivable rate uplink}}, is omitted.
\end{proof}
\vspace{+1em}
\textbf{\emph{Asymptotic Analysis:}}
In order to present further insights into the system performance, we now investigate the asymptotic outage probability and achievable rate by neglecting the interference terms in the UL and DL SINRs. Therefore, with negligible LI effect, one can omit the term $P_{\AP} h_{\AP\AP}$ in~\eqref{eq:SINR at AP in uplink}. Finally, the cdf of the SNR in the special case of $\alpha=2$ can be obtained as
\begin{align} \label{eqn:cdf_SNRa_Asyp}
F_{\SNRa}(z) &=1-\left(1+\frac{z}{\psi_{\Us}}\right)^{-1}e^{-\frac{\lambda_{\Ds}\pi d^2}{1 + \frac{\psi_{\Us}}{z}}},
\end{align}
where $\psi_{\Us}=\frac{P_{\Us}}{\Sn}\lambda_{\Ds}\pi$. The asymptotic outage of the UL transmission can be determined by substituting $z=\gamma_{th}$ into~\eqref{eqn:cdf_SNRa_Asyp}. Furthermore, the corresponding asymptotic spatial average capacity of the UL user is
\begin{align} \label{eqn:achievabe rate UP}
R_{\AP} &=\frac{1}{\log2}
\left(\frac{1}{\psi_{\Us}}-1\right)^{-1}e^{-\frac{\lambda_{\Ds}\pi d^2 }{\psi_{\Us}}}\nonumber\\
&\quad\times
\left(
\mathrm{Ei}\left(\frac{\lambda_{\Ds}\pi d^2}{1-\psi_{\Us}}\right)-
\mathrm{Ei}\left(\frac{\psi_{\Us}\lambda_{\Ds}\pi d^2}{1-\psi_{\Us}}\right)\right),
\end{align}
where $\mathrm{Ei}\left(\cdot\right)$ is the exponential integral function defined in~\cite[Eq. (8.211.1)]{Integral:Series:Ryzhik:1992}.
\begin{figure}[th]
\vspace{0.6em}
\centering
\includegraphics[width=87mm, height=64mm]{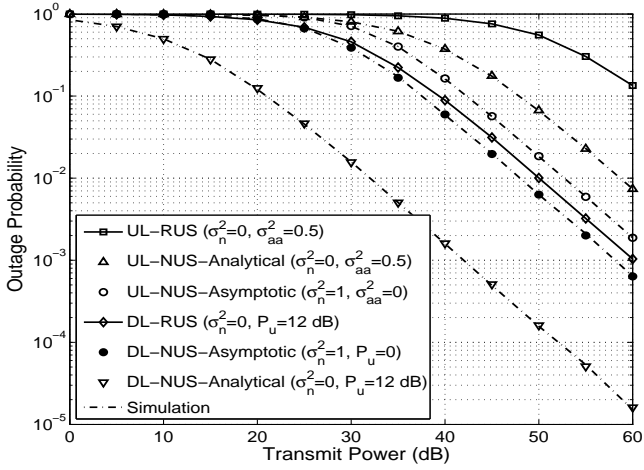}
\vspace{-1.7em} \caption{Outage performance of the DL and UL user for nearest user selection (NUS) and random user selection (RUS) ($d=25$ m and $\gamma_{th}=3$ dB)} \label{fig:outage_probability_DL_user} \vspace{-1.2em}
\end{figure}
Similarly, by neglecting the term $P_{\Us} h_{ \Us\Ds} \ell( x_{\Us}-x_{\Ds})$ in~\eqref{eq:SINR: downlonk user:single antenna case},  a valid assumption for $P_{\Us} d^{-\alpha}\ll1$, we  obtain
\begin{align} \label{eqn:cdf_SNRd_Asyp}
 F_{\SNRd}(z) &=
\left\{%
 \begin{array}{clcr}
  1-\left(1+\frac{z\lambda_{\Ds}\pi}{\psi_{\Ds}}\right)^{-1}&                  \alpha=2,  \\
  1- \sqrt{\frac{\psi_{\Ds}}{2z}} e^{\frac{\psi_{\Ds}}{8z}}D_{-1} \left( \sqrt{\frac{\psi_{\Ds}}{2z}}\right)                 &     \alpha=4
  \end{array}%
\right.
\end{align}
where $\psi_{\Ds} = \frac{P_{\AP}}{\Sn}(\lambda_{\Ds}\pi)^2 $ and $D_{-1}(\cdot)$ denotes a Parabolic cylinder function~\cite[Eq. (9.241.2)]{Integral:Series:Ryzhik:1992}. Accordingly, the DL user asymptotic outage probability can be readily obtained by substituting $z=\gamma_{th}$ into~\eqref{eqn:cdf_SNRd_Asyp}. Moreover, the corresponding rates are given by
\vspace{-0.4em}
\begin{align} \label{eqn:achievabe rate DL}
R_{\Ds} &=\frac{1}{\log2}
\left\{%
 \begin{array}{clcr}
  \left(\frac{\lambda_{\Ds}\pi}{\psi_{\Ds}}-1\right)^{-1}\log\left(\frac{\lambda_{\Ds}\pi}{\psi_{\Ds}}\right)&                  \alpha=2,  \\
  \int_{0}^{\infty}\frac{1}{z+1}\sqrt{\frac{\psi_{\Ds}}{2z}} e^{\frac{\psi_{\Ds}}{8z}}D_{-1} \left( \sqrt{\frac{\psi_{\Ds}}{2z}}\right)dz              &     \alpha=4.
  \end{array}%
\right.
\end{align}
\subsection{Half-Duplex Mode}
In this subsection, we compare the performance of the HD and FD modes of operation at the AP. In the HD mode of AP operation, AP employs orthogonal time slots to serve the DL and UL user, respectively. In order to keep our comparisons fair, we consider \emph{``antenna conserved''} (AC) and \emph{``RF-chain conserved''} (RC)
scenarios. Under AC condition, the total number of antennas used by the HD AP and FD AP is kept identical. However, the number of radio frequency (RF) chains employed by the HD AP is twice that of the FD AP~\cite{Khojastepour:Mobicom:2012} and hence former system would be a costly option. Under RC condition, the total number RF chains used is same for the HD and FD modes. Therefore, in any transmission (UL or DL), the HD AP only uses a single antenna under the RC condition, while it uses two antennas under the AC condition.
\begin{figure}[th]
\vspace{-0.3em}
\centering
\includegraphics[width=87mm, height=67mm]{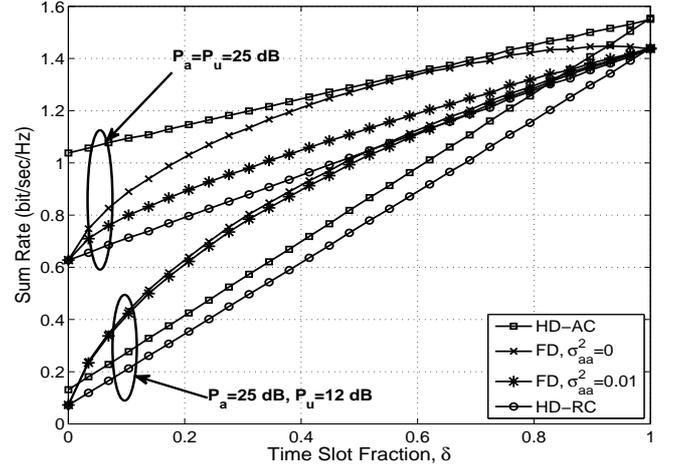}
\vspace{-1.9em} \caption{Average sum rate versus $\delta$ for the FD (FD) and half-dulpex (HD) AP  ($\alpha=2$, and $d=25$ m).}
\label{fig: sum_rate_v_delta} \vspace{-0.6em}
\end{figure}
The average sum rate under the RC condition can be expressed as
\vspace{-0.2em}
\begin{align}\label{eq: sum rate of single-antenna HD AP}
\RHDs=&\delta\E{\log_2\left(1+\snr_{\Ds}\ell(x_{\Ds})| h_{\AP,\Ds}|^2\right)}\nonumber\\
&+(1-\delta)\E{\log_2\left(1+\snr_{\Us}\ell(x_{\Us})| h_{\Ds,\AP}|^2\right)},
\end{align}
where $\delta$ ($0<\delta<1$) is a fraction  of the time slot duration of $T$, used for DL transmission, $\snr_{\Ds} = P_{\AP}^{\HD}/\Sn$, and $\snr_{\Us} = P_{\Us}^{\HD}/\Sn$.

Under the AC condition, using the weight vector $\vw_{\MRC} = \vh_{\Ds,\AP}^H$ for the maximum ratio combining (MRC) receiver, and the maximum ratio transmission (MRT) precoding vector $\vw_{\MRT} = \frac{\vh_{\AP,\Ds}^H} {\|\vh_{\AP,\Ds}\|}$, the average achievable rate can be obtained as
\vspace{-0.2em}
\begin{align}\label{eq: sum rate of dual-antenna HD AP}
\RHDd =& \delta\E{\log_2\left(1+\!\frac{\snr_{\Ds} }{2}\ell(x_{\Ds})\|\vh_{\AP,\Ds}\|^2\right)}\nonumber\\
&+(1-\delta)\E{\log_2\left(1+\!\snr_{\Us} \ell(x_{\Us})\| \vh_{\Ds,\AP}\|^2\right)}.
\end{align}
\section{Numerical Results and Discussion}\label{sec:Numerical results}
Here, we investigate the system performance  and confirm the derived analytical results through comparison with Monte Carlo simulations. We evaluate the performance in a cell of radius $R_c=200$ m and for $\lambda_{\Ds}=1\times 10^{-3}$ node/$\text{m}^2$. Moreover, with curves shown in Figs. 3-5, we assume that the total power of  the AP and UL user for  FD and HD modes is the same.

Fig.~\ref{fig:outage_probability_DL_user} shows the outage probability versus SNR for the nearest DL user (to the AP) and UL user for $\alpha=2$, $d=25$ m and $\gamma_{th}=3$ dB. In this figure, the X-axis indicates the power of the transmitter (i.e., AP for DL and UL user for UL). The outage probability of the RUS scheme is also included as a benchmark comparison. The `Analytical' curves are plotted from~\eqref{eq: cdf of SINRd integral over r: alpha 2 Final} and~\eqref{eq: final cdf of SIR DL} with $z=\gamma_{th}$, for nearest UL user and DL user, respectively, which clearly match the Monte Carlo simulated curves. As expected, we see that the nearest user selection (NUS) scheme outperforms the RUS scheme. In addition, the `Asymptotic' curves plotted from~\eqref{eqn:cdf_SNRa_Asyp} and~\eqref{eqn:cdf_SNRd_Asyp} tightly converge to the simulation values.

In Fig.~\ref{fig: sum_rate_v_delta} we compare the average sum rate as a function of $\delta$ for the FD and HD operation and for two different values of $\Sap$. We assume same total energy consumption for both FD and HD operation and plot the sum rate for two different power constraints $(P_{\AP},P_{\Us}) =(25~\text{dB}, 25~\text{dB} )$ (symmetric) and $(P_{\AP},P_{\Us}) =(25~\text{dB}, 12~\text{dB})$ (asymmetric). In particular, numerical results lead to the following conclusions:
\emph{$1$)} As expected, the sum rate under the RC condition is worse than those of other scenarios. \emph{$2$)} In the asymmetric case, FD operation outperforms HD within the practical range of $\delta$. However, in the symmetric case, AC condition achieves the best performance even in case of perfect LI cancellation (i.e., $\Sap=0$).
\emph{$3$)} The symmetric case is more vulnerable to the LI power (Please see Fig.~\ref{fig: Average_spectral_efficiency_gain}).

\begin{figure}[t]
\vspace{-0.35em}
\centering
\includegraphics[width=87mm, height=65mm]{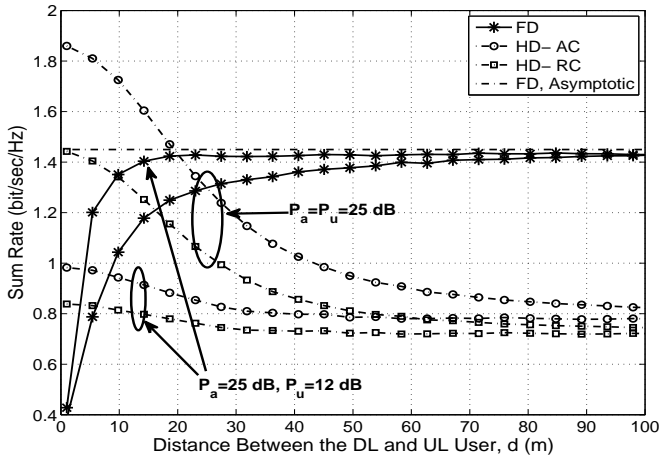}
\vspace{-1.3em} \caption{ Average sum rate versus $d$ the FD and half-dulpex (HD) AP ($\alpha=2$, $\delta=0.5$, and $\Sap=0.1$).} \label{fig: sum_rate_v_d} \vspace{-1.1em}
\end{figure}
\begin{figure}[t]
\vspace{-.2em}
\centering
\includegraphics[width=87mm, height=65mm]{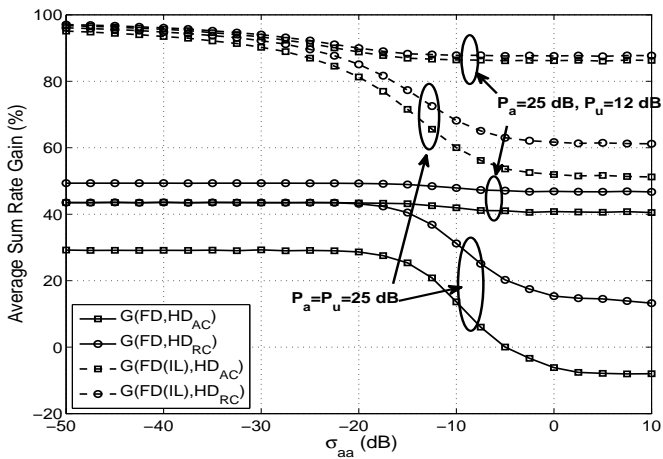}
\vspace{-1.9em} \caption{Average sum rate gain of the system ($\alpha=2$, $d=25$ m, and $\delta=0.5$).}
\label{fig: Average_spectral_efficiency_gain} \vspace{-1.5em}
\end{figure}
In Fig.~\ref{fig: sum_rate_v_d}, we present the average sum rate (with $\delta = 0.5$ and $\Sap=0.1$) versus the distance $d$ between the UL and DL user achieved by the FD and HD modes of operation. There are two main observations that can be extracted from this figure. First, the sum rate shows the opposite behaviors in FD and HD modes as $d$ increases. This result can be explained as follows. Notice that when $d$ increases the inter user interference between the DL and UL user decreases and thus $\SINRd$ and consequently the average sum rate of the FD system increases. On the other hand, the sum rate of the HD operation is inversely proportional to $d$. Therefore, increasing $d$, reduces the sum rate. Secondly, as $d$ increases the average sum rate of the FD system converges to the sum of asymptotic rate in~\eqref{eqn:achievabe rate UP} and~\eqref{eqn:achievabe rate DL}.

In Fig.~\ref{fig: Average_spectral_efficiency_gain} we plot the average sum rate gain, which is defined as $G(\FD,\HD_{i})\! = (R_{\FD}\!-R_{\HD}^i)/R_{\FD}$ versus $\sigma_{\AP\AP}$ and for $d=25$ m and $\delta=0.5$. The sum rate gain of the interference-limited FD system  is also included for comparison (dashed line curves). A general observation is that FD  significantly outperforms the HD counterpart when LI is substantially suppressed. However, when $\sigma_{\AP\AP}\geq -5$ dB, the AC-HD system outperforms the FD system (symmetric power case). Moreover, the symmetric power case is more sensitive to the LI effect.


\section{Conclusion}
We have analyzed the performance of a wireless network scenario where a FD AP is communicating with spatially random HD user terminals in the downlink and uplink channels simultaneously. We derived the outage probability and achievable sum rate of the system, considering the impact of the LI channel and inter user interference. Then, we compared the performance of the FD and HD modes of operations for the same total power budget. We found that even if the LI cancellation is imperfect, FD transmissions with different transmit power levels at the AP and the UL user can achieve significant performance gains as compared to the HD mode of operation.
\balance

\bibliographystyle{IEEEtran}

\end{document}